\documentclass[11pt]{article}

\usepackage{amsthm,amsfonts,amsmath,amssymb,epsfig,color,float,graphicx,verbatim, enumitem, url}
\usepackage[margin=1in]{geometry}

\newtheorem{theorem}{Theorem}[section]

\newtheorem{lemma}[theorem]{Lemma}
\newtheorem{informal theorem}[theorem]{Theorem (informal statement)}

\newtheorem{proposition}[theorem]{Proposition}
\newtheorem{corollary}[theorem]{Corollary}

\newtheorem{fact}[theorem]{Fact}

\theoremstyle{definition}
\newtheorem{definition}[theorem]{Definition}

\newcommand{\dtv}{d_{\mathrm {TV}}}

\newcommand{\var}{\mathbf{Var}}

\newcommand{\cov}{\mathbf{Cov}}
\newcommand{\eps}{\epsilon}
\newcommand{\tr}{\mathrm{tr}}
\newcommand{\sym}{\mathrm{Sym}}

\newcommand{\R}{\mathbb{R}}

\newcommand{\poly}{\mathrm{poly}}
\newcommand{\E}{\mathbb{E}}

\newcommand{\bP}{\mathbf{P}}
\newcommand{\bQ}{\mathbf{Q}}
\newcommand{\bp}{\mathbf{P}}
\newcommand{\bq}{\mathbf{Q}}
\newcommand{\br}{\mathbf{R}}

\newcommand{\supp}{\mathrm{supp}}

\title{Robust Learning of Mixtures of Gaussians}

\author{
Daniel M. Kane\thanks{Supported by NSF Award CCF-1553288 (CAREER) and a Sloan Research Fellowship.}\\
University of California, San Diego\\
{\tt dakane@cs.ucsd.edu}\\
}

\begin{document}

\maketitle

\begin{abstract}
We resolve one of the major outstanding problems in robust statistics. In particular, if $X$ is an evenly weighted mixture of two arbitrary $d$-dimensional Gaussians, we devise a polynomial time algorithm that given access to samples from $X$ an $\eps$-fraction of which have been adversarially corrupted, learns $X$ to error $\poly(\eps)$ in total variation distance.
\end{abstract}

\section{Introduction}
\subsection{Background}

A Gaussian mixture is a probability distribution given as a convex combination of Gaussian distributions. Namely, a mixture of $k$ Gaussians is a probability distribution of the form $\bP = \sum_{i=1}^k w_i N(\mu_i,\Sigma_i)$ where $w_i$ are positive real valued weights that sum to $1$ and $\mu_i$, $\Sigma_i$ are the component means and covariance matrices. Gaussian mixture models are arguably the most important latent variable model with their study dating back over a century to Pearson \cite{pearson}.

Given the importance of these mixture models, it is natural to consider the problem of trying to learn an unknown mixture from samples. There is a long line of work on trying to solve this problem with different notions of error and different assumptions on the underlying mixture. A large number of papers (\cite{sep1,sep2,sep3,sep4,sep5,sep6}) have studied the problem of learning such mixtures under various separation assumptions between the components. In relatively recent work, \cite{LearningMixturesWithoutNoise1} gave the first efficient algorithm for learning mixtures of two Gaussians in parameter distance. This was later improved to a nearly optimal algorithm by \cite{LearningMixturesWithoutNoise2}. Another pair of works (\cite{mixtureparameters1,mixtureparameters2}) generalize this result to mixtures of $k$ Gaussians. There have also been a number of papers that have studied the related problem of density estimation for such mixture (\cite{density1,density2,density3,mixtureparameters1,LearningMixturesWithoutNoise2,density4}).

The field of robust statistics (see \cite{RobustStatistics1,RobustStatistics2}) attempts to understand which statistical estimation tasks can still be accomplished in the presence of outliers. While the information-theoretic aspects of many of these problems have been well understood for some time, until recently all known estimators were either computationally intractable or had error rates that scaled poorly with dimension.

It was only recently that the first works on computationally efficient, high dimensional robust statistics (\cite{dkk,lrs}) overcame this obstacle. These papers provided the first techniques giving computationally efficient algorithms that learn high dimensional distributions such as Gaussians and product distributions to small error in total variational distance even in the presence of a constant fraction of adversarial noise. Since the publication of these works, there has been an explosion of results in this area with a number of papers expanding upon the original techniques, finding ways to make these algorithms more efficient and finding ways to apply them to new classes of distributions. For a survey of the recent work in this area see \cite{RobustSurvey}.

Since the inception of robust computational statistics, the problem of robustly learning a mixture of even two arbitrary Gaussians has remained a major open problem. Although a number of works (for example \cite{dkk,clustering1,clustering2}) have solved special cases of this problem, the general case has remained illusive. In this paper, we resolve this problem, providing the first efficient, robust algorithm for learning an arbitrary (equal weight) mixture of two Gaussians.

\subsection{Our Results}

In order to introduce our results, we will first need to define the error model:
\begin{definition}[Strong Contamination Model]
We say that an algorithm has access to $\eps$-noisy samples from a distribution $X$ if the algorithm can access the following oracle a single time:

The algorithm picks a number $N$, then $N$ i.i.d. samples from $X$ are generated: $x_1,\ldots,x_N$. An adversary is then allowed to inspect these samples and replace at most $\eps N$ of them with arbitrarily chosen new samples. The algorithm is then given the list of (modified) samples.
\end{definition}

We note that this is often referred to as the strong adversary model and is the strongest of the contamination models commonly studied in robust statistics. We also note that while as stated the algorithm can only make a single call to this oracle, it is not hard to see that at the cost of slightly increasing $\eps$, it can simulate any polynomial number of calls simply by asking for a larger number of samples and randomly partitioning these samples into smaller subsets. Since the adversary does not know ahead of time what the partition is going to be, it will be unlikely that any subset will have more than $2\eps$ corrupted samples (at least assuming that the number of samples in each part is large relative to $1/\eps$ and the number of parts).

We also note that if $X$ and $X'$ are distributions that are guaranteed to have total variational distance at most $\eps$, then an algorithm can simulate (sufficiently many) $3\eps$-noisy samples from $X'$ given access to $\eps$-noisy samples from $X$. This is because a sample from $X$ can be thought of as a sample from $X'$ that is corrupted with a probability of at most $\eps$. Given a sufficiently large number of samples $N$, then with high probability at most $2\eps N$ of these samples will be corrupted in this way.

In these terms, our main theorem is easy to state:

\begin{theorem}\label{mainTheorem}
Let $G_1$ and $G_2$ be arbitrary Gaussians in $\R^d$. There exists an algorithm that given $\eps$-noisy samples to $X=(G_1+G_2)/2$ runs in time $\poly(d/\eps)$ and with probability at least $2/3$ returns a distribution $\hat X$ so that $\dtv(X,\hat X) = \poly(\eps)$.
\end{theorem}

\subsection{Comparison with Prior Work}

As mentioned before, \cite{LearningMixturesWithoutNoise1,LearningMixturesWithoutNoise2} show how to learn a mixture of two Gaussians without noise. These works learn the \emph{parameters} of the mixture to small error in time and samples polynomial in $d$ and $\sigma$, where $\sigma$ is the variance of the mixture. We note that this notion of parameter distance is somewhat different than the total variational distance considered in our work. However, it is not hard to show that this notion of parameter distance is stronger. In particular, the algorithms in these papers can be used to learn a mixture of two Gaussians to $\eps$-error in total variational distance in polynomial time (though this reduction is not entirely trivial). That being said, these algorithms hold only in the non-robust setting and fail very quickly when even a small amount of noise is introduced. Furthermore, while the parameter distance metric might be considered more powerful than the total variational distance metric, the latter is more natural to consider when looking at robust learning. In particular, it is easy to see that in most settings it is impossible even information-theoretically to learn $X$ to better than $\eps$ error in total variational distance when only given access to $\eps$-corrupted samples.

In terms of robust learning of mixtures of Gaussians, only limited results were known until this point. In \cite{dkk} it was shown how to learn a single Gaussian robustly to small error in total variational distance. That paper also showed how to learn mixtures of $k$ \emph{identity covariance} Gaussians in polynomial time for any constant $k$. More recently, \cite{clustering1} and \cite{clustering2} independently showed how to learn mixtures of $k$ Gaussians robustly under the assumption that the component Gaussians were highly separated in total variational distance. Our paper for the first time solves this problem in nearly full generality. It learns an arbitrary mixture of two Gaussians robustly under only the assumption that the weights are equal.

\subsection{Techniques}

It is shown in \cite{LearningMixturesWithoutNoise2} that learning sixth moments of a mixture of two Gaussians is sufficient to uniquely identify them. At a high level this will also be our strategy for learning the mixture. However, there are several problems with this strategy.

To begin with, even learning moments of a distribution at all is non-trivial in the presence of adversarial errors. A single corrupted sample can already change the empirical moments of a distribution by arbitrarily much. Fortunately, the robust statistics literature has figured out how to get around this in many cases. In particular, it has been known for some time how one can estimate the mean of a random variable $X$ robustly given that one knows that the covariance of $X$ is bounded. If one wants to learn higher moments of a distribution, one can apply this statement to learn an approximation to the mean of $X^{\otimes k}$, assuming that the covariance of this random variable is bounded.

Of course boundedness will also be a problem for us. A mixture of arbitrary Gaussians will have no a priori bounds on its covariance. Thus, an important first step will be to normalize the mixture. In particular, we need a way to approximate the mean and covariance matrix of $X$, so that by applying an appropriate affine transformation we can reduce to the problem where $X$ is close to mean $0$ and identity covariance.

It was shown in \cite{dkk} how to robustly learn an unknown covariance Gaussian to small total variational distance error (which corresponds the learning the covariance matrix in terms of a relative Frobenius norm metric). This result requires that we have a particular kind of relationship hold between the second moments of $X$ and the fourth moments of $X$. This unfortunately, does not hold for arbitrary mixtures of Gaussians, but we will show that it holds for mixtures where the components are not too far apart from each other in total variational distance.

Fortunately, we can use the result of \cite{clustering2} to learn our mixture in the case where the two components are substantially separated in total variational distance. This allows us to reduce to the case where the two components are relatively close, which is sufficient to allow us to perform the normalization procedure described above.

The next obstacle comes in estimating the higher moments of $X$ once we have normalized it. We know that we can estimate the moments of $X$ to small error given that the covariance of $X^{\otimes k}$ is bounded. Unfortunately, even under good circumstances, this is unlikely to be the case. Instead, we need to replace $X^{\otimes k}$ by an appropriate tensor of Hermite polynomials. This, we can show will have bounded covariance assuming that we are in the case where the individual components of $X$ are not too far separated.

Finally, given estimates of the higher moments of $X$ we need to be able to recover the individual components to relatively small error. This is helped by the assumption that our components are not too far separated, as it means that the distance to which we need to learn the parameters of the individual components is not too bad. For example, if one component had much smaller covariance than the other, than we might be forced to learn the parameters of that component to much higher precision in order to guarantee a small error in total variational distance. However, even learning the parameters from an approximation of the moments is non-trivial. A method is given in \cite{LearningMixturesWithoutNoise2} that does this by considering a number of one-dimensional projections, however, this technique will lose dimension-dependent factors, which we cannot afford. Instead we devise a new technique that involves random projections of higher moment tensors into lower dimensions, and doing some guessing to remove some low rank noise.

\subsection{Structure of the Paper}

We begin in Section \ref{backgroundSec} with some basic notation and results that will be used throughout the rest of the paper. In Section \ref{separatedSec}, we deal with the special case where the component Gaussians have small overlap. Then in Section \ref{reductionSec}, we show how if this is not the case we can reduce to the situation where $X$ is mean $0$ and identity covariance. Once we have done this, Section \ref{momentsSec} shows how we can robustly compute moments of $X$ and how to use them to approximate the individual components. Next, in Section \ref{combineSec}, we show how to combine everything and prove Theorem \ref{mainTheorem}. Finally, in Section \ref{conclusionsSec}, we discuss some ideas on how to further extend these results.

\section{Background}\label{backgroundSec}

\subsection{Notation}

We will use $[n]$ to denote the set $\{1,2,\ldots,n\}$. For weights $w_i\geq 0$ summing to $1$ and $\bp_i$ probability distributions, we use $\bq= \sum w_i\bp_i$ to denote the probability distribution given by their mixture. In particular, the probability density function $\bq(x)$ is given by the average of the density functions, $\bq(x) = \sum w_i \bp_i(x)$.

\subsection{Distance Between Gaussians}

We will need the following approximation of the variational distance between two Gaussians:
\begin{fact}\label{GaussianTVFact}
$$
\dtv(N(\mu_1,\Sigma_1),N(\mu_2,\Sigma_2)) = O(((\mu_1-\mu_2)\cdot \Sigma_1^{-1} (\mu_1-\mu_2))^{1/2}+\|\Sigma_1^{-1/2}\Sigma_2 \Sigma_1^{-1/2}-I \|_F ).
$$
\end{fact}
\begin{proof}
The cases where $\mu_1=\mu_2$ or $\Sigma_1=\Sigma_2$ are proved in \cite{dkk}. The full result follows from noting that
$$
\dtv(N(\mu_1,\Sigma_1),N(\mu_2,\Sigma_2)) \leq \dtv(N(\mu_1,\Sigma_1),N(\mu_2,\Sigma_1))+\dtv(N(\mu_2,\Sigma_1),N(\mu_2,\Sigma_2)).
$$
\end{proof}

\subsection{Tensors}

For our purposes an $m$-tensor will be an element in $(\R^d)^{\otimes m} \cong \R^{d^m}$. This can be thought of as a vector with $m^d$ coordinates. These coordinates however, instead of being indexed $1,2,\ldots,m^d$ are index by $m$-tuples of integers from $1$ to $d$. We will often used $T_{i_1i_2\ldots i_m}$ to denote the coordinate of the $m$-tensor $T$ indexed by the $m$-tuple $(i_1,i_2,\ldots,i_m)$. By abuse of notation, we will also sometimes use this to denote the entire tensor. This allows us to make use of Einstein summation notation. Given an $m_1$-tensor $A$ and an $m_2$-tensor $B$ written as a product with $k$ of their indices in common, such as:
$$
A_{i_1i_2\ldots i_k j_{k+1}\ldots j_{m_1}}B_{i_1i_2\ldots i_k \ell_{k+1}\ldots \ell_{m_2}},
$$
This represents the $(m_1+m_2-2k)$-tensor $C$ given by summing over the shared indices. In particular,
\begin{align*}
& C_{j_{k+1}\ldots j_{m_1} \ell_{k+1}\ldots \ell_{m_2}}\\
= & \sum_{i_1,i_2,\ldots, i_k=1}^dA_{i_1i_2\ldots i_k j_{k+1}\ldots j_{m_1}}\cdot B_{i_1i_2\ldots i_k \ell_{k+1}\ldots \ell_{m_2}}.
\end{align*}
The special case where $k=0$ (there are not shared indices) we abbreviate as $C=A\otimes B$. Similarly, we write $A^{\otimes t}$ to denote the $t$-fold tensor power of $A$ with itself. In a further abuse of notation, we will at times associate vectors in $\R^d$ with $1$-tensors and matrices in $\R^{d\times d}$ with $2$-tensors.

Another important special case is where $m_1=m_2=k$. In this case $C$ is a $0$-tensor, or equivalently a single real number given by the sum of the products of the corresponding entries of $A$ and $B$. We call this the dot product of $A$ and $B$, which we denote $A\cdot B$ or $\langle A,B\rangle.$

\subsubsection{Tensor Norms}

It will be important for us to bound various norms of tensors. Perhaps the most fundamental such norm is the $L^2$ or Frobenius norm. In particular, given a tensor $A$, we denote by $\|A\|_F$ the square root of the sum of the squares of the entries of $A$. Equivalently, $\|A\|_F = \sqrt{A\cdot A}.$

Another relevant tensor norm involves a relationship between tensors and matrices. In particular, if $T$ is an $m$-tensor and $S=\{j_1,j_2,\ldots,j_k\}$ is a subset of $[m]$ we note that if $A$ is a $k$-tensor, the product
$$
T_{i_1i_2\ldots i_m}A_{i_{j_1}i_{j_2}\ldots i_{j_k}}
$$
is an $(m-k)$-tensor. This allows us to interpret $T$ as a linear transformation from the space of $k$-tensors to the space of $(m-k)$-tensors, or equivalently as a $d^k\times d^{m-k}$ matrix. We will (assuming the subset $S$ is clear from context) use $\|T\|_2$ to denote the largest singular value of this matrix.

\subsubsection{Symmetric Tensors}

Many of the tensors that we will be working with will have a large degree of symmetry. Taking advantage of this will allow us to simplify some of our formulas. To be specific if $T$ is an $m$-tensor $\pi$ is a permutation of $[m]$ we define the $m$-tensor $\pi T$ by
$$
(\pi T)_{i_1i_2\ldots i_m} = T_{i_{\pi(1)}i_{\pi(2)}\ldots i_{\pi(m)}}.
$$
We define
$$
\sym(T) = \frac{1}{m!}\sum_{\pi \in S_m} \pi T.
$$
Furthermore, if $P$ is any partition of the set $[m]$, use $\sym_P(T)$ to denote the above but averaged only over permutations $\pi$ that preserve $P$.

\subsection{Robust Statistics}

We will need a couple of basic results in robust statistics. We begin with one of the most basic results in the area, namely that an algorithm with access to noisy samples from a distribution with bounded covariance can efficiently approximate the mean of the distribution.

\begin{theorem}\label{basicCovLearnerTheorem}
Let $X$ be a distribution on $\R^d$ with $\cov(X)\leq I$, then there exists a polynomial time algorithm which given $\eps$-noisy samples from $X$ returns a $\hat\mu$ so that $|\E[X]-\hat\mu|\leq O(\sqrt{\eps})$ with high probability.
\end{theorem}

There are many proofs of Theorem \ref{basicCovLearnerTheorem}. For completeness, and in order to assist with our next result we provide one below.
\begin{lemma}\label{filterLemma}
Let $\bP$ be a discrete probability distribution in $\R^d$ with $\cov(\bP)\leq 1$ and let $\bQ$ a discrete measure on $\R^d$ so that $|\bP-\bQ|\leq \eps$ for a sufficiently small constant $\eps$. Let $\bQ'=\bQ/|\bQ|_1$ be the normalization of $\bQ$ to a probability distribution. Then either:
\begin{itemize}
\item $\cov(\bQ') \leq 3I$ in which case $|\E[\bP]-\E[\bq']| = O(\sqrt{\eps}).$
\item There exists an algorithm which given $\bq$ runs in polynomial time and returns a measure $\bq_0$ so that $|\bp-\bq_0| \leq \eps$ and $|\supp(\bq_0)|<|\supp(\bq)|$.
\end{itemize}
\end{lemma}

Theorem \ref{basicCovLearnerTheorem} follows from Lemma \ref{filterLemma} by letting $\bp$ be the empirical distribution over the uncorrupted samples (perhaps scaled down slightly so that the covariance is less than $I$) and starting with $\bq$ as the empirical distribution over the noisy samples handed to the algorithm. It is clear that $|\bp-\bq|_1\leq 2\eps$ and that $\E[\bp]$ is close to $\E{X}$. The algorithm them iteratively applies Lemma \ref{filterLemma} at each step finding a new distribution whose support is smaller and smaller, eventually terminating at a distribution $\bq$ so that $\E[\bq']$ is sufficiently close to $\E[\bp]$.

We now prove Lemma \ref{filterLemma}

\begin{proof}
We begin with a proof of the first statement. If $\cov(\bq')\leq 3I$, we note that $\dtv(\bp,\bq') = O(\eps)$. Thus, for some distributions $\br,\bp_0,\bq_0$ and $\delta=2\eps$ we can write $\bp=(1-\delta)\br+\delta\bp_0$ and $\bq'=(1-\delta)\br+\delta \bq_0$. We claim that $|\E[\bp]-\E[\br]|=O(\sqrt{\delta})$ and similarly that  $|\E[\bq']-\E[\br]|=O(\sqrt{\delta})$. The final result will follow from the triangle inequality. For the latter statement (the former follows similarly), we note that
\begin{align*}
\cov(\bq') & = (1-\delta) \cov(\br) + \delta\cov(\bq_0) +\delta(1-\delta) (\E[\br]-\E[\bq_0])(\E[\br]-\E[\bq_0])^T \\ & \geq \delta/2(\E[\br]-\E[\bp_0])(\E[\br]-\E[\bp_0])^T.
\end{align*}
Since $\cov(\bq')\leq B I$, this implies that $|\E[\br]-\E[\bq_0]| = O(\sqrt{B/\delta}).$ Given that $\E[\bq'] = \E[\br]+\delta(\E[\bq_0]-\E[\br])$,  we have $|\E[\bp]-\E[\bq']| = O(\sqrt{B\eps})$.

For the latter result, we assume that $\cov(\bq')$ has largest eigenvalue $B\geq 3$. We find a unit vector $v$ so that $\var(v\cdot \bq') \geq 0.9 B$. We define a function $f(x)$ on $\R^d$ by
$$
f(x) = (v\cdot (x - \E[\bq']))^2.
$$
We define $\bq_0$ by letting $\bq_0(x)=0$ for $x$ not in the support of $\bq$ and otherwise letting $\bq_0(x) = \bq(x)(1-cf(x))$, where $1/c$ is the maximum value of $f(x)$. It is clear that $\bq_0$ has smaller support than $\bq$. It remains to show that it is closer to $\bp$. We begin by comparing the amount of mass lost to the amount that would have been lost if $\bq$ were equal to $\bp$.

Note that $\E[f(\bq')] = \var(v\cdot \bq') \geq 0.9B$. On the other than
$$\E[f(\bp)]=\var(v\cdot \bp)+(v\cdot(\E[\bq']-\E[\bp]))^2 \leq 1 + O(B\eps). $$
Where we use the bound on $|\E[\bq']-\E[\bp]|$ from above. Note that if $B\geq 3$ and $\eps$ sufficiently small we have that $|\bq|_1\E[f(\bq')] \geq 2\E[f(\bp)].$ This is enough. In particular, let $\bp_0$ be $\bp$ with the probability mass at $x$ decreased by a $(1-cf(x))$ factor for each $x$. We note that since this operation keeps the sign of each coordinate the same,
\begin{align*}
|\bp_0-\bp_0|_1 & = |\bp-\bq|_1-|(\bp-\bp_0)-(\bq-\bq_0)|_1\\
& \leq |\bp-\bq|_1 - |\bq-\bq_0|_1 + |\bp-\bp_0|_1.
\end{align*}
This implies
$$
|\bp-\bq_0|_1 \leq |\bp-\bq|_1 - |\bq-\bq_0|_1 + 2|\bp-\bp_0|_1.
$$
However, it is easy to see that $|\bp-\bp_0|_1 = c\E[f(\bp)]$ and $|\bq-\bq_0|_1 = c|\bq|_1 \E[f(\bq')].$ Combining with the above inequality completes our proof.

\end{proof}

We will also need an algorithm for learning the covariance of a random variable under appropriate conditions. The following is a generalization of the argument from \cite{dkk} for learning the covariance matrix of a Gaussian (note here that $X$ is standing in for the random variable $GG^T$).

\begin{theorem}\label{CovarianceLearnerTheorem}
Let $X$ be a distribution on $\R^{d\times d}$, where $X$ is supported on the subset of $\R^{d\times d}$ corresponding to the symmetric, positive semi-definite matrices. Suppose that $\E[X]=\Sigma$ and that for any symmetric matrix $A$ we have that $\var(\tr(AX)) = O(\sigma^2\|\Sigma^{1/2} A \Sigma^{1/2}\|_F^2).$ Then there exists a polynomial time algorithm that given sample access to an $\eps$-corrupted version of $X$ for $\eps$ less than a sufficiently small multiple of $\sigma^{-2}$ returns a matrix $\hat{\Sigma}$ so that with high probability $\|\Sigma^{-1/2}(\Sigma -\hat{\Sigma})\Sigma^{-1/2}\|_F = O(\sigma\sqrt{\eps}).$
\end{theorem}
\begin{proof}
We being by reducing to the case where $\Sigma^{-1/2} X \Sigma^{-1/2}$ is bounded. It is easy to see from the above bounds that $\Sigma^{-1/2} X \Sigma^{-1/2}$ has covariance bounded by $\sigma^2 I_{d^2}$. This implies that it is only with at most $\eps$ probability that $|\Sigma^{-1/2} X \Sigma^{-1/2}| \geq d^2\sigma^2/\eps$. Replacing $X$ by $X'$, the conditional distribution on this event not happening, we note that $\dtv(X,X')\leq \eps$ and so this difference can be thought of as merely increasing our noise rate by $\eps$. Furthermore, the bounded covariance implies that removing an $\eps$-probability event changes the expectation of $\Sigma^{-1/2} X \Sigma^{-1/2}$ by at most $O(\sigma \sqrt{\eps})$, and so will not change the correctness of our approximation. Furthermore, although the removal of these extreme samples might decrease the covariance of $X$, it cannot increase it substantially. This shows that it suffices for our algorithm to work for the bounded variable $X'$. We henceforth assume that $X$ is bounded in this way.

Let $\bp$ be the uniform distribution over the uncorrupted samples. Assuming that we took sufficiently many samples, it is easy to see that with high probability the following hold:
\begin{itemize}
\item $\|\Sigma^{-1/2}(\E[\bp]-\E[X])\Sigma^{-1/2}\|_F  = O(\sigma\sqrt{\eps}).$
\item For every matrix $A$, $\var(\tr(A\bp)) = O(\sigma^2\|\Sigma^{1/2} A \Sigma^{1/2}\|_F^2).$
\end{itemize}
We assume throughout the following that the above hold.

The algorithm is given a discrete measure $\bq$, the uniform distribution over the noisy samples. It is the case that $|\bp-\bq|_1 \leq 2\eps$. The algorithm will iteratively produce a sequence of such measures each with $|\bp-\bq_i|_1 \leq 2\eps$, and with smaller and smaller support until it eventually returns a hypothesis $\hat\Sigma$.

We begin by letting $\bp_0$ be the measure given by the pointwise minimum of $\bp$ and $\bq$. We note that $\bp_0$ is obtained from $\bp$ by removing $O(\eps)$ mass. Therefore, since $\Sigma^{-1/2} \bp \Sigma^{-1/2}$ has covariance $O(\sigma^2 I)$, we have that $\|\Sigma^{-1/2} (\E[\bp]-\E[\bp_0/|\bp_0|_1])\Sigma^{-1/2}\|_F = O(\sigma \sqrt{\eps})< 1/2.$ In particular, this means that $\Sigma^{-1/2}\E[\bp_0/|\bp_0|_1]\Sigma^{-1/2}\geq I/2.$ Since $\E[\bq'] \geq \E[\bp_0/|\bp_0|_1],$ we have that $\E[\bq'] \geq \Sigma/2$.

In particular, this means that if $C$ is a sufficiently large constant that
$$
\cov(\E[\bq']^{-1/2} \bp \E[\bq']^{-1/2}/(C \sigma^2)) \leq \cov(2\Sigma^{-1/2} \bp \Sigma^{-1/2}/(C \sigma^2)) \leq I_{d^2}.
$$
This allows us to apply Lemma \ref{filterLemma} to the distribution $\E[\bq']^{-1/2} \bp \E[\bq']^{-1/2}/(C \sigma^2)$ and the measure $\E[\bq']^{-1/2} \bq \E[\bq']^{-1/2}/(C \sigma^2)$ as they have $L^1$ distance at most $2\eps$ and the former has covariance bounded by the identity. This either gives us a new measure $\bq_0$ with smaller support that is not too far from $\bp$. Or it is the case that
$$
\|\E[\bq']^{-1/2} (\E[\bp] - \E[\bq']) \E[\bq']^{-1/2}\|_F = O(\sigma \sqrt{\eps}) < 1/2.
$$
Calling $\E[\bp]=\Sigma_0\approx \Sigma$ and $\E[\bq']=\hat\Sigma$ we have
$$
\|\hat\Sigma^{-1/2}(\Sigma_0-\hat\Sigma)\hat\Sigma^{-1/2}\|_F < 1/2.
$$
Therefore $\hat\Sigma \geq \Sigma_0/2 \geq \Sigma/3$. Therefore,
$$
\|\Sigma^{-1/2}(\Sigma-\hat\Sigma)\Sigma^{-1/2}\|_F = O(\sigma \sqrt{\eps}),
$$
as desired.
\end{proof}

\subsection{Moment Computations}

We will also be working heavily with higher moments of Gaussians and will need to perform some basic computations about them.

\begin{proposition}\label{momentsProp}
Let $G=N(\mu,\Sigma)$ be a Gaussian in $\R^d$ then
$$
\E[G^{\otimes m}]_{i_1\ldots i_m} = \sum_{\substack{\textrm{Partitions }P\textrm{ of }[m]\\ \textrm{ into sets of size 1 and 2}}} \bigotimes_{\{a,b\}\in P} \Sigma_{i_a,i_b} \bigotimes_{\{c\}\in P} \mu_{i_c}.
$$
\end{proposition}
\begin{proof}
We note that it is enough to show that both sides are equal after dotting them with $v^{\otimes m}$ for any vector $v$. The left hand side becomes $\E[(v\cdot G)^m]$. We note that $v\cdot G$ is a Gaussian with mean $v\cdot \mu$ and variance $v^T \Sigma v$. Letting $H$ be the standard Gaussian, this yields $\E[(v\cdot \mu + \sqrt{v^T \Sigma v}H)^m]$. Expanding with the binomial theorem yields
$$
\sum_{k=0}^m \binom{m}{k} (v^T \Sigma v)^{k/2} (v\cdot \mu)^{m-k} \E[H^k] = \sum_{\ell=0}^{\lfloor m/2 \rfloor} \binom{m}{2\ell} (v^T \Sigma v)^{\ell} (v\cdot \mu)^{m-2\ell} (2\ell-1)!!
$$

When we dot the right hand side with $v^{\otimes m}$ on the other hand, each term in the sum contributes $(v^T \Sigma v)^\ell (v\cdot \mu)^{m-2\ell}$ where $\ell$ is the number of pairs in the partition $P$. The number of such partitions with exactly $\ell$ pairs is $\binom{m}{2\ell}(2\ell-1)!!$. This is because there are $\binom{m}{2\ell}$ ways to choose which elements are in the $\ell$ pairs and once that is decided, $(2\ell-1)!!$ ways to pair them up. Summing over $\ell$ gives the same expression as the above, proving our proposition.
\end{proof}

Unfortunately, the higher moments will often be difficult to compute directly (at least in the robust setting). Instead, we will need to get at them indirectly through a slightly different set of ``moments''. The following tensors correspond to the standard multivariate Hermite polynomials.

\begin{definition}
Define the degree-$m$ Hermite polynomial tensor as
$$
h_m(x) :=  \sum_{\substack{\textrm{Partitions }P\textrm{ of }[m]\\ \textrm{ into sets of size 1 and 2}}} \bigotimes_{\{a,b\}\in P} (-I_{i_a,i_b}) \bigotimes_{\{c\}\in P} x_{i_c}.
$$
\end{definition}

\begin{lemma}\label{HermiteExpectationLem}
If $G=N(\mu,I+\Sigma)$ then
$$
\E[h_m(G)] = \sum_{\substack{\textrm{Partitions }P\textrm{ of }[m]\\ \textrm{ into sets of size 1 and 2}}} \bigotimes_{\{a,b\}\in P} \Sigma_{i_a,i_b} \bigotimes_{\{c\}\in P} \mu_{i_c}.
$$
\end{lemma}
\begin{proof}
By the definition of the Hermite tensor and Proposition \ref{momentsProp}, we have that
\begin{align}
& \ \E[h_m(G)] \nonumber\\ & = \sum_{\substack{\textrm{Partitions }P\textrm{ of }[m]\\ \textrm{ into sets of size 1 and 2}}} \bigotimes_{\{a,b\}\in P} (-I_{i_a,i_b})\E\left[ \bigotimes_{\{c\}\in P} x_{i_c}\right]\label{expectationPartitionEquation}\\
& = \sum_{\substack{\textrm{Partitions }P\textrm{ of }[m]\\ \textrm{ into sets of size 1 and 2}}} \sum_{\substack{\textrm{Partitions }Q\textrm{ of }S\\\textrm{where }S\textrm{ is the set of singletons in }P \\\textrm{ into sets of size 1 and 2}}} \bigotimes_{\{a,b\}\in P} (-I_{i_a,i_b})\bigotimes_{\{a,b\}\in Q} (I+\Sigma)_{i_a,i_b} \bigotimes_{\{c\} \in Q} \mu_{i_c}.\nonumber
\end{align}
Combining the two sums, we note that this is equivalent to partitions $R$ of $[m]$ into sets of size $1$ and $2$ with the sets of size $2$ being marked as type $1$ (coming from $P$) or type $2$ (coming from $Q$). Thus, this is
$$
\sum_{\substack{\textrm{Marked partitions }R\textrm{ of }[m]\\ \textrm{ into sets of size 1 and 2}}}\bigotimes_{\{a,b\}\in R \textrm{ type 1}} (-I_{i_a,i_b})\bigotimes_{\{a,b\}\in R \textrm{ type 2}} (I+\Sigma)_{i_a,i_b} \bigotimes_{\{c\} \in Q} \mu_{i_c}.
$$
However, if we fix $R$ and sum over the choices for each part of size $2$ of whether it is type $1$ or type $2$, we get
$$
\sum_{\substack{\textrm{Partitions }R\textrm{ of }[m]\\ \textrm{ into sets of size 1 and 2}}}\bigotimes_{\{a,b\}\in R} ((I+\Sigma)-I)_{i_a,i_b} \bigotimes_{\{c\} \in Q} \mu_{i_c},
$$
which is easily seen to be equal to the desired quantity.
\end{proof}

Finally, we will also need to bound the covariance of the Hermite polynomial tensors. For this we have the following lemma.

\begin{lemma}\label{HermiteVarianceLem}
If $G=N(\mu,I+\Sigma)$ then $\E[h_m(G)\otimes h_m(G)]$ equals
$$
\sum_{\substack{\textrm{Partitions }P\textrm{ of }[2m]\\ \textrm{ into sets of size 1 and 2}}} \bigotimes_{\substack{\{a,b\}\in P\\a,b\textrm{ in same half of }[2m]}} \Sigma_{i_a,i_b} \bigotimes_{\substack{\{a,b\}\in P\\a,b\textrm{ in different halves of }[2m]}} (I+\Sigma)_{i_a,i_b}\bigotimes_{\{c\}\in P} \mu_{i_c}.
$$
\end{lemma}
\begin{proof}
The proof is essentially the same as that of Lemma \ref{HermiteExpectationLem}. The primary difference is that in our version of Equation \eqref{expectationPartitionEquation} we will only allow $P$ to contain pairs that do not cross between different halves of $[2m]$. This means that for the partition $R$, only pairs that do not cross can be type 1, and thus these pairs contribute $(I+\Sigma)$ rather than just $\Sigma.$
\end{proof}

\subsection{Tournaments}

In order to learn our mixture of Gaussians in full generality, we will need to have different algorithms for different cases and will need to make several correct guesses in order to succeed. By considering all possible combinations of guesses, the algorithm will end up with a number of hypothesis distributions at least one of which is guaranteed to be close to the true one. From this point, we will need to run a tournament in order to find a hypothesis that is not too far away. This is by now a fairly standard procedure in learning theory, though we need to verify here that this can be done even with only access to $\eps$-noisy samples.

\begin{lemma}\label{tournamentLemma}
Let $X$ be an unknown distribution and let $H_1,\ldots,H_n$ be distributions with explicitly computable probability density functions that can be efficiently sampled from. Assume furthermore than $\min_{1\leq i\leq n}(\dtv(X,H_i)) \leq \eps$. Then there exists an efficient algorithm that given access to $\eps$-noisy samples from $X$ along with $H_1,\ldots,H_n$ computes a $1\leq m\leq n$ so that with high probability
$$
\dtv(X,H_m) = O(\eps).
$$
\end{lemma}
\begin{proof}
For each $i,j$ define the set $A_{i,j}$ to be the set of point where the probability density of $H_i$ is bigger than the probability density of $H_j$. We note in particular that $\dtv(H_i,H_j)=|H_i(A_{i,j})-H_j(A_{i,j})|$. Taking enough samples from $X$, we can ensure that with high probability $X(A_{i,j})$ is within $\eps$ of the fraction of the uncorrupted samples lying in $A_{i,j}$ for each $i,j$. Note that this will imply that $P^{(0)}_{i,j}$, the fraction of the noisy samples lying in $A_{i,j}$, is within $2\eps$ of $X(A_{i,j})$.

Additionally, for each $m$ we sample enough samples from $H_m$ to compute an approximation $P^{(m)}_{i,j}$ to $H_m(A_{i,j})$ to enough accuracy so that with high probability $|P^{(m)}_{i,j}-H_m(A_{i,j})| \leq \eps$ for all $i,j,m$.

Our algorithm then returns any $m$ so that $|P^{(m)}_{i,j}-P^{(0)}_{i,j}|\leq 4\eps$ for all $i,j$. This will necessarily exist because
$$
|P^{(m)}_{i,j}-P^{(0)}_{i,j}| \leq |P^{(m)}_{i,j}-H_m(A_{i,j})| + |P^{(0)}_{i,j}-X(A_{i,j})| + \dtv(H_m,X),
$$
which is at most $4\eps$ for any $m$ for which $\dtv(H_m,X)\leq \eps$.

However, such an $m$ will be sufficient this is because
$$
|P^{(m)}_{i,j}-P^{(0)}_{i,j}| \geq |X(A_{i,j})-H_m(A_{i,j})| - |P^{(m)}_{i,j}-H_m(A_{i,j})| - |P^{(0)}_{i,j}-X(A_{i,j})|.
$$
On the other hand, we also have that
$$
|X(A_{i,j})-H_m(A_{i,j})| \geq |H_{k}(A_{i,j}) - H_m(A_{i,j})| - \dtv(H_k,X).
$$
If we take $H_k$ so that $\dtv(H_k,X)\leq \eps$ and take $i=k$ and $j=m$, we have that $|H_{k}(A_{i,j}) - H_m(A_{i,j})|= \dtv(H_k,H_m)$, and combining this with the above, we have
$$
|P^{(m)}_{k,m}-P^{(0)}_{k,m}| \geq \dtv(H_m,H_k)-3\eps \geq \dtv(X,H_m) - 4\eps.
$$
Therefore, given any $m$ with $|P^{(m)}_{i,j}-P^{(0)}_{i,j}|\leq 4\eps$ for all $i,j$ will have $\dtv(X,H_m)\leq 8\eps$.

This completes our proof.
\end{proof}

\section{Separated Gaussians}\label{separatedSec}

We note that if $G_1$ and $G_2$ are separated in variational distance that our problem is already solved by work of \cite{clustering2}.
\begin{theorem}\label{separatedCaseTheorem}
If $G_1$ and $G_2$ are Gaussians with $\dtv(G_1,G_2)>1-\poly(\eps)$, then there is a polynomial time algorithm that given access to $\eps$-noisy samples from $X=(G_1+G_2)/2$ learns $X$ to error $\poly(\eps)$.
\end{theorem}

We can henceforth assume that $G_1$ and $G_2$ have total variation distance at most $1-\delta$ with $\delta= \eps^c$ for some small positive constant $c$. We would like to know what this entails.
\begin{lemma}\label{separationLem}
Suppose that $G_1$ and $G_2$ are Gaussians with total variation distance at most $1-\delta$. Let $X=(G_1+G_2)/2$ have covariance $\Sigma$. Then we have that:
\begin{enumerate}
\item $\cov(G_1),\cov(G_2) \gg \delta^2 \Sigma$
\item $\|\Sigma^{-1/2} (\cov(G_1)-\cov(G_2)) \Sigma^{-1/2}\|_F = O(\log(1/\delta))$
\end{enumerate}
\end{lemma}
\begin{proof}
We note that this statement is invariant under affine transformations, by applying such a transformation, we can assume that $\Sigma=I$. We will proceed by contradiction. In particular, we will show that if either of the above are violated, then $\dtv(G_1,G_2)> 1-\delta$. We also assume throughout that $\delta$ is sufficiently small.

For the first condition, assume that (without loss of generality) $\cov(G_1)$ has an eigenvalue less than $c\delta^3$ for a sufficiently small constant $c$. In particular this means that there is a unit vector $v$ so that $\var(v\cdot G_1) \leq c\delta^2.$ We will show that $\dtv(G_1,G_2)\geq \dtv(v\cdot G_1,v\cdot G_2) \geq 1-\delta$. We note that $v\cdot G_1$ and $v\cdot G_2$ are one dimensional Gaussians whose mixture has unit variance. Let $\Delta$ be the distance between the means and let $\sigma$ be the variance of $v\cdot G_2$. We note that $\var(v\cdot X) = \sigma^2/2 + \Delta^2/4 + \var(v\cdot G_1)/2$. In particular this implies that either $\sigma \gg 1$ or $|\Delta|\gg 1.$ In fact if $\sigma \geq \delta$, it is easy to see that $\dtv(v\cdot G_1,v\cdot G_2) \gg 1-O(\sqrt{c}\delta)$ since the standard deviations of these Gaussians differ by a factor of at least $\sqrt{c}\delta$. Otherwise, if $\sigma < \delta$ then $|\Delta| \gg 1$ and the components are separated by $\Omega(1/\sqrt{\delta})$ standard deviations, which implies that the total variational distance is at least $1-\delta$.

For the second condition, we will use some results from \cite{SQ}. For Gaussians $G_1,G_2$ they define $h_\Sigma(G_1,G_2)$ where $\dtv(G_1,G_2) \geq 1-\exp(-\Omega(h_\Sigma(G_1,G_2))).$ They also show that if $A=\cov(G_1), B=\cov(G_2)$ we have that
$$
h_\Sigma(G_1,G_2) = \Theta\left(\sum_{\substack{\lambda \textrm{ eigenvalues of }B^{-1/2}AB^{-1/2}}} \min(|\log(\lambda)|,|\log(\lambda)|^2) \right).
$$
We note that if $\dtv(G_1,G_2)\leq 1-\delta$, then $h_\Sigma(G_1,G_2) = O(\log(1/\delta)).$

Note that if we apply an orthogonal change of variables so that $A$ and $B$ are simultaneously diagonalized, that the eigenvalues of $B^{-1/2}AB^{-1/2}$ are just the ratios of the eigenvalues of $A$ with the corresponding eigenvalues of $B$. We note that since $\min(|\log(x)|,|\log(x)^2|) \geq \Omega(((1-x)/(1+x))^2)$, we have that
\begin{align*}
h_\Sigma(G_1,G_2) & \geq \Omega\left( \sum_{\substack{\lambda \textrm{ eigenvalues of }(A+B)^{-1/2}(A-B)(A+B)^{-1/2}}} \lambda^2 \right) \\ & = \Omega(\|(A+B)^{-1/2}(A-B)(A+B)^{-1/2}\|_F^2).
\end{align*}
However, if $I=\cov(X) \geq (\cov(G_1)+\cov(G_2))/2$. We have that $A+B \leq 2I$, and so the above is at least $\Omega(\|\cov(G_1)-\cov(G_2)\|_F^2)$, and so $\|\cov(G_1)-\cov(G_2)\|_F$ must be $O(\log(1/\delta))$.
\end{proof}

\section{Covariance Approximation for Non-Separated Mixtures}\label{reductionSec}

Our next result shows allows us to robustly estimate the covariance of $X$ assuming that the component Gaussians are not too far apart.

\begin{proposition}\label{mixtureCovarianceExtimationProp}
Let $X=(G_1+G_2)/2$ be a mixture of Gaussians with $\dtv(G_1,G_2)\leq 1-\delta$. Let $\Sigma = \cov(X)$. There exists a polynomial time algorithm that given $\eps$-noisy samples from $X$ for $\eps$ less than a sufficiently small constant returns a hypothesis $\hat{\Sigma}$ so that
$$
\| \Sigma^{-1/2}(\hat{\Sigma}-\Sigma)\Sigma^{-1/2}\|_F \leq O(\log(1/\delta)\sqrt{\eps}).
$$
\end{proposition}
\begin{proof}
Let $Y=X-X'$ be the difference of independent copies of $X$. Note that our algorithm has sample access to a $2\eps$-noisy samples of $Y$ by subtracting pairs of $\eps$-noisy samples from $X$. Let $Z= YY^T$. Note that $\cov(X)$ is proportional to $\E[Z]$. It therefore suffices to show that $Z$ satisfies the hypotheses of Theorem \ref{basicCovLearnerTheorem}.

Let $\Sigma=\cov(X)$, since this problem is invariant under linear change of variables, so we assume for convenience that $\Sigma=I$. Note that this problem is unaffected by translating $X$ to have mean $0$. Let $G_i = N(\mu_i,\Sigma_i)$ where $\mu_1 = \mu = -\mu_2$. We note that $\mu\mu^T \ll I$. We also note that by Lemma \ref{separationLem} that $I \gg \Sigma_i \gg \delta^2 I$ and $\|\Sigma_1-\Sigma_2\|_F = O(\log(1/\delta)).$

We note that $Y$ is a mixture of $G_1-G_1'$, $G_1-G_2'$, $G_2-G_1'$ and $G_2-G_2'$ (where the primed versions of the variables are independent copies). Hence $Z$ is a mixture of $Z_i=D_iD_i^T$ where each $D_i$ is one of these four distributions. Since the covariance of a mixture of distributions is the mixture of the covariances plus the covariance of the distribution over the component means, we need to show that:
\begin{enumerate}
\item For each of these distributions $Z_i$, and every matrix $A$ we have that $\var(\tr(AZ_i)) = O(\|A\|_F^2)$.
\item For any two of these distributions $Z_i$ and $Z_j$, we have that $\|\E[Z_i-Z_j]\|_F^2 = O(\log^2(1/\delta)).$
\end{enumerate}
To show this we note that each $D_i$ is a Gaussian $N(\mu^\ast_i,\Sigma^\ast_i)$ with $\Sigma^\ast_i = O(I)$, $|\mu^\ast_i|=O(1)$ and $\|\Sigma^\ast_i-\Sigma^\ast_j\|_F =O(\log(1/\delta))$. The first claim above follows by noting that of $D=N(\mu^\ast,\Sigma^\ast)$ that $\var(D\otimes D)$ is
\begin{align*}
& \E[D^{\otimes 4}] - \E[D^{\otimes 2}]^{\otimes 2}\\
= \ \ & \Sigma^\ast_{i_1i_2}\Sigma^\ast_{i_3i_4} +  \Sigma^\ast_{i_1i_3}\Sigma^\ast_{i_2i_4} + \Sigma^\ast_{i_2i_3}\Sigma^\ast_{i_1i_4}
\\ + & \Sigma^\ast_{i_1i_2}\mu^\ast_{i_3}\mu^\ast_{i_4} + \Sigma^\ast_{i_1i_3}\mu^\ast_{i_2}\mu^\ast_{i_4} + \Sigma^\ast_{i_1i_4}\mu^\ast_{i_2}\mu^\ast_{i_3} + \Sigma^\ast_{i_2i_3}\mu^\ast_{i_1}\mu^\ast_{i_4} + \Sigma^\ast_{i_2i_4}\mu^\ast_{i_1}\mu^\ast_{i_3} + \Sigma^\ast_{i_3i_4}\mu^\ast_{i_1}\mu^\ast_{i_2}\\
+ & \mu^\ast_{i_1}\mu^\ast_{i_2}\mu^\ast_{i_3}\mu^\ast_{i_4} - \Sigma^\ast_{i_1i_2}\Sigma^\ast_{i_3i_4} - \Sigma^\ast_{i_1i_2}\mu^\ast_{i_3}\mu^\ast_{i_4} - \Sigma^\ast_{i_3i_4}\mu^\ast_{i_1}\mu^\ast_{i_2} - \mu^\ast_{i_1}\mu^\ast_{i_2}\mu^\ast_{i_3}\mu^\ast_{i_4} \\
= \ \ & \mathrm{Sym}_{\{\{1,2\},\{3,4\}\}}(2\Sigma^\ast_{i_1i_3}\Sigma^\ast_{i_2i_4}+4\Sigma^\ast_{i_1i_2}\mu^\ast_{i_3}\mu^\ast_{i_4}).
\end{align*}
It is enough to show that for each of these terms $T$ in the above that $A_{i_1i_2}T_{i_1i_2i_3i_4}A_{i_3i_4} = O(\|A\|_F^2).$ For $T=\Sigma^\ast_{i_1i_2}\mu^\ast_{i_3}\mu^\ast_{i_4}$, this is $(A\mu^\ast)^T \Sigma^\ast (A\mu^\ast)$, and it follows since $|A\mu^\ast| \leq |\mu^\ast|\|A\|_F$ and $\|\Sigma^\ast\|_2 = O(1)$. For $T=\Sigma^\ast_{i_1i_3}\Sigma^\ast_{i_2i_4}$ this follows from $A_{i_1i_2}T_{i_1i_2i_3i_4}A_{i_3i_4} = \|(\Sigma^\ast)^{1/2} A(\Sigma^\ast)^{1/2} \|_F^2 = O(\|A\|_F^2)$ since $\Sigma^\ast = O(I)$.

For the second note, we have that $\E[Z_i]=\Sigma_i^\ast +\mu_i^\ast \otimes \mu_i^\ast$. The result follows from $|\mu_i^\ast|=O(1)$ and $\|\Sigma_i^\ast -\Sigma_j^\ast\|_F^2 = O(\log^2(1/\delta)).$

This completes our proof.
\end{proof}

This allows us to learn an approximation to the covariance of $X$. By applying a linear transformation, we can make this covariance approximately the identity. Thus, using Theorem \ref{basicCovLearnerTheorem}, we can approximate the mean of $X$ to translate it into standard form. Using this, we can reduce to the case where $X$ has mean $0$ and covariance $I$.
\begin{proposition}\label{standardFormReductionProp}
Let $X=(G_1+G_2)/2$ where $G_i$ are $d$-dimensional Gaussians with $\dtv(G_1,G_2)\leq 1-\delta$. There exists an algorithm that given $\eps$-noisy samples from $X$ with $\eps < \delta^3$ runs in polynomial time and with high probability returns an invertible affine transformation $L$ so that $L(X)$ is $O(\sqrt{\eps}\log(1/\delta)/\delta^2)$-close in total variation distance to a distribution $X'=(G_1'+G_2')/2$ where $G_1',G_2'$ are Gaussians with $\dtv(G_1',G_2')\leq 1-\delta$ and $\E[X']=0$ and $\cov(X')=I$.
\end{proposition}
\begin{proof}
We begin by using Proposition \ref{mixtureCovarianceExtimationProp} to learn a $\hat{\Sigma}$ so that $\| \Sigma^{-1/2}(\hat{\Sigma}-\Sigma)\Sigma^{-1/2}\|_F \leq O(\log(1/\delta)\sqrt{\eps}).$ It is then the case that $X_1 :={\hat \Sigma}^{-1/2} X$ is a mixture of Gaussians ${\hat \Sigma}^{-1/2} G_1, {\hat \Sigma}^{-1/2}G_2$ with covariance $O(\sqrt{\eps}\log(1/\delta))$-close to $I$ in Frobenius norm. Since $\cov(X_1)=O(I)$, and since we have sample access to $\eps$-corrupted samples from $X_1$, using Theorem \ref{basicCovLearnerTheorem}, we can estimate $\E[X_1]$ to error $O(\sqrt{\eps})$ giving us a value $\hat \mu$.

We now define $L(x) = {\hat \Sigma}^{-1/2}x-\hat \mu$. It is clear that applying $L$ to $X$ gives a random variable with mean close to $0$ and covariance close to the identity. This means that there is an affine transformation $M(x)=Ax+b$ so that $M(X_1)$ has mean $0$ and identity covariance with $\|A-I\|_F=O(\sqrt{\eps}\log(1/\delta))$ and $|b|=O(\sqrt{\eps})$. We next note that $X_1 = (L(G_1)+L(G_2))/2$. We claim that $L(G_i)$ is $O(\sqrt{\eps}\log(1/\delta)/\delta)$-close to $M(L(G_i))$. This would imply that $X_1$ was $O(\sqrt{\eps}\log(1/\delta)/\delta)$-close to $M(L(X))=X'$, which clearly satisfies our hypotheses.

To do this we note that $I \gg \cov(L(G_i)) \gg \delta^2 I$ and the $|\E[L(G_i)]|=O(1)$. From this, and the bounds on $A-I$ and $b$, we can infer that $\|\cov(L(G_i))-\cov(M(L(G_i)))\|_F = O(\sqrt{\eps}\log(1/\delta))$ and $|\E[L(G_i)] - \E[M(L(G_i))]| = O(\sqrt{\eps}\log(1/\delta))$. The result now follows from Fact \ref{GaussianTVFact}.
\end{proof}

Thus, assuming that the Gaussians $G_i$ are not $(1-\delta)$-separated, we can apply Proposition \ref{standardFormReductionProp} to reduce to the case where the $G_i$ are not separated and where $\cov(X)=I$ and $\E[X]=0$ (assuming we are willing to replace our error by $\eps'=O(\sqrt{\eps}\log(1/\delta)/\delta^2)$). We note that in this case we can write $G_1 = N(\mu,I-\mu\mu^T+\Sigma)$ for some $\mu$ and $\Sigma$. In this case, $G_2$ would then be $N(-\mu,I-\mu\mu^T-\Sigma)$. Since $G_1$ and $G_2$ are not too far separated, $\|\Sigma\|_F = O(\log(1/\delta))$. Also, by the lack of separation, we have that $\cov(G_i) \gg \delta^2 I$. Therefore, if we can learn $\mu$ and $\Sigma$ to $L^2$ slash Frobenius error $\eta$, then we can learn $G_i$ to error $O(\eta/\delta^2)$.

Thus, from here on out, we will assume that $G_1 = N(\mu,I-\mu\mu^T+\Sigma)$ and $G_2 = N(-\mu,I-\mu\mu^T-\Sigma)$ with $\|\Sigma\|_F = O(\log(1/\delta))$. We will show an algorithm that given access to $\eps$-corrupted samples to $X$, makes polynomially many guesses at least one of which is likely to be within $\eta$ of $(\mu,\Sigma).$

\section{Moment Estimation}\label{momentsSec}

In this section, we show that we can compute the higher moments of $X$ in the above situation and discuss what that means.
\begin{lemma}
Let $X=(G_1+G_2)/2$ with $G_1 = N(\mu,I-\mu\mu^T+\Sigma)$ and $G_2 = N(-\mu,I-\mu\mu^T-\Sigma)$. Then
%$$
%E[h_3(X)] = \sym(3\mu\otimes \Sigma).
%$$
$$
E[h_4(X)] = \sym(3\Sigma^{\otimes 2}-2\mu^{\otimes 4}).
$$
$$
E[h_6(X)] = 16 \mu^{\otimes 6}.
$$
\end{lemma}
\begin{proof}
We note that $\E[h_6(X)] = (\E[h_6(G_1)]+\E[h_6(G_2)])/2$. By Lemma \ref{HermiteExpectationLem}, this is half of
\begin{align*}
& \sym(15(\Sigma-\mu^{\otimes 2})^{\otimes 3}+45 \mu^{\otimes 2}(\Sigma-\mu^{\otimes 2})^{\otimes 2}+15\mu^{\otimes 4}(\Sigma-\mu^{\otimes 2})+\mu^{\otimes 6} )\\
+ & \sym(15(-\Sigma-\mu^{\otimes 2})^{\otimes 3}+45 \mu^{\otimes 2}(-\Sigma-\mu^{\otimes 2})^{\otimes 2}+15\mu^{\otimes 4}(-\Sigma-\mu^{\otimes 2})+\mu^{\otimes 6} )
\end{align*}
We note that the expression in the second line is obtained from that in the first by negating every $\mu$ and $\Sigma$ term. This means that all terms cancel out except for those that are tensor products of an even number of $\mu$'s and $\Sigma$'s. This leaves only $\mu^{\otimes 2}\Sigma^{\otimes 2}$ and $\mu^{\otimes 6}$ terms. A careful accounting of the number of each term left yields the desired answer.

We note that $\E[h_4(X)] = (\E[h_4(G_1)]+\E[h_4(G_2)])/2$. By Lemma \ref{HermiteExpectationLem}, this is half of
\begin{align*}
& \sym(3(-\mu\otimes \mu+\Sigma)^{\otimes 2}+6\mu\otimes\mu\otimes (-\mu\otimes \mu+\Sigma)+\mu^{\otimes 4})\\+ & \sym(3(-\mu\otimes \mu-\Sigma)^{\otimes 2}+6(-\mu)\otimes(-\mu)\otimes (-\mu\otimes \mu-\Sigma)+(-\mu)^{\otimes 4}).
\end{align*}
We note that the expression in the second line is obtained from that in the first by negating every $\mu$ and $\Sigma$ term. This means that all terms cancel out except for those that are tensor products of an even number of $\mu$'s and $\Sigma$'s. This leaves only $\Sigma^{\otimes 2}$ and $\mu^{\otimes 4}$. A careful count of the number of copies of each gives the stated result.
\end{proof}

To show that we can compute these moments, we need to know that the covariance of $h_m(X)$ is bounded so that we can apply Theorem \ref{basicCovLearnerTheorem}.
\begin{lemma}
In the above situation
$$
\left\|\cov(h_m(X)) \right\|_2 = O_m(1+\|\Sigma\|_F+|\mu|)^{2m}.
$$
\end{lemma}
\begin{proof}
We note that by definition $\cov(h_m(X))$ is upper bounded by $T=\E[h_m(X)\otimes h_m(X)]$, where this $2m$-tensor is thought of as a matrix over $m$-tensors acting by multiplying the first $m$ entries. We note that by Lemma \ref{HermiteVarianceLem} that $T$ is a sum of terms each of which are a product of copies of $\mu,\Sigma$ and $I$ where the copies of $I$ all cross between the first $m$ entries and the last $m$. We claim that any individual term of this form has operator norm $O(1+\|\Sigma\|_F+|\mu|)^{2m}$. For such a term $T'$, we consider the size of $T'_{i_1\ldots i_{2m}}A_{i_1\ldots i_m}$ for some $m$-tensor $A$. First, we consider the effect of multiplying $A$ by the copies of $I$ in $T'$. Since these terms always have one coordinate in the first $m$ and one in the last $m$ this, corresponds to multiplying $A$ by the identity in some coordinate, and thus does not affect the $L^2$ norm. We have at most $2m$ other terms that are all copies of $\Sigma$ or $\mu$ and multiplying by them each increases the norm of the resulting matrix by a factor of at most $\|\Sigma\|_F+|\mu|$. This completes the proof.
\end{proof}

Combining the above with Theorem \ref{basicCovLearnerTheorem} we obtain the following:
\begin{corollary}\label{momentEstimationCorollary}
Given $X$ as above there exists a polynomial time algorithm that given access to $\eps$-noisy samples from $X$ computes $\sym(6\Sigma^{\otimes 2}-4\mu^{\otimes 4})$ to Frobenius error $O(\sqrt{\eps}(1+\|\Sigma\|_F+|\mu|)^4),$ and $\mu^{\otimes 6}$ to Frobenius error $O(\sqrt{\eps}(1+\|\Sigma\|_F+|\mu|)^6).$
\end{corollary}

Now that we can approximate these tensors, we want to show that we can use them to approximate $\mu$ and $\Sigma$. We begin by showing that we can approximate $\mu$ from an approximation of $\mu^{\otimes 6}$.
\begin{proposition}\label{muEstimationProp}
Let $\mu$ be a vector and $\eta>0$ a parameter. There exists a polynomial time algorithm that given a $6$-tensor $M$ with
$$
\|M-\mu^{\otimes 6}\|_F \leq \eta
$$
computes a vector $\hat \mu$ so that with probability at least $1/3$ $|\hat\mu-\mu| = O(\eta^{1/6})$.
\end{proposition}
\begin{proof}
Thinking of $M$ as a $d\times d^5$ matrix, we note that it is $\eta$-close to the rank $1$ matrix $\mu(\mu^{\otimes 5})^T$. Letting $M'$ be the closest rank $1$ approximation to $M$ (computed via a singular value decomposition), we have that $M' = u v^T$ is $O(\eta)$-close to $\mu(\mu^{\otimes 5})^T$. From here it is easy to see that either $u$ or $-u$ must be $O(\eta/|\mu|^5)$-close to $\mu$. Therefore, either $\mu$ is $O(\eta^{1/6})$-close to $0$ or $O(\eta^{1/6})$ close to $u$ or $-u$. Guessing which case we are in gives an appropriate answer with probability $1/3$.
\end{proof}

We note that given this and out approximation to the $4^{th}$ moment, we can approximate $\sym(\Sigma^{\otimes 2})$ to error $\tilde O(\eps^{1/12})$.

We have left to show that knowing this tensor is sufficient to learn $\Sigma$.
\begin{proposition}\label{SigmaEstimationProp}
Let $\Sigma$ be a symmetric matrix with $\|\Sigma\|_F \leq 1$, and $\eta>0$ a parameter. There exists an algorithm that given a $4$-tensor $M$ with
$$
\|M-\sym(\Sigma^{\otimes 2})\|_F \leq \eta
$$
runs in polynomial time and with probability at least $\poly(\eta)$ returns $\hat{\Sigma}$ so that
$$
\|\hat{\Sigma}-\Sigma\|_F = O(\eta^c)
$$
for some positive constant $c > 0$.
\end{proposition}
\begin{proof}
We begin by reducing to the case where $\|\Sigma\|_F=1$.

We note that by guessing $s=\|\Sigma\|_F$ to error $\eta$, we can divide $M$ by $s^2$ to get an $\eta/s^2$ approximation to $\sym((\Sigma/s)^{\otimes 2})$. If $s\leq \eta^{1/3}$, we can take $\hat\Sigma = 0$. Otherwise, we have an $O(\eta^{1/3})$-approximation of $\sym(\Sigma'^{\otimes 2})$ for some $\Sigma'$ of Frobenius norm $1$. If we can solve the problem in this case, finding a $\hat\Sigma$, so that $\|\hat\Sigma -\Sigma'\|_F = O((\eta/s^2)^c)=O(\eta^{c/3})$, we can solve the original problem by returning $s\hat\Sigma.$

We next split into cases based upon whether $\Sigma$ is $\eta^{1/2}$-close in Frobenius norm to a rank-$3$ matrix.

If $\Sigma$ is close, there are some vectors $u,v,w$ so that $\|\Sigma-(uu^T+vv^T+ww^T)\|_F = O(\eta^{1/2})$. Let $\Sigma'=uu^T+vv^T+ww^T$. Note that if we treat $\sym(\Sigma'^{\otimes 2})$ as a $d\times d^3$ matrix that it is rank at most $3$. Thus, $M$ (when treated as a $d\times d^3$ matrix) is also $O(\eta^{1/2})$ close to a rank $3$ matrix. Let $M'$ be the closest rank $3$ approximation to $M$ (obtained by a singular value decomposition). It is easy to see that $\|M'-\sym(\Sigma'^{\otimes 2})\|=O(\eta^{1/2})$. Let $V$ be the span of the singular vectors (on the $\R^d$ side) of $M'$. We claim that all of $u,v,w$ are close to $V$. In particular, if $z$ is a unit vector orthogonal to $V$, it is not hard to see that $\langle M' , z^{\otimes 4}\rangle = 0$ while
$$
\langle \sym(\Sigma'^{\otimes 2}),z^{\otimes 4} \rangle = \langle \Sigma'^{\otimes 2},z^{\otimes 4} \rangle = ((z\cdot u)^2+(z\cdot v)^2+(z\cdot w)^2)^2.
$$
On the other hand, the difference in these is at most $\|z^{\otimes 4}\|_F \|M'-\sym(\Sigma'^{\otimes 2})\|_F = O(\eta^{1/2})$. Therefore, for any such $z$, $|z\cdot u|,|z\cdot v|,$ and $|z\cdot w|$ are all $O(\eta^{1/8})$. Since this holds for all such $z$, this means that $u,v,w$ are all within $O(\eta^{1/8})$ of lying in $V$. Our algorithm can guess $\eta$-approximations to their projections onto $V$ (there are only $\poly(1/\eta)$ many possibilities), and if it succeeds, return $\hat\Sigma =\hat u \hat u^T+\hat v \hat v^T+\hat w \hat w^T$, which will be within $O(\eta^{1/8})$ of the true $\Sigma$.

Next we assume that $\Sigma$ is not $\eta^{1/2}$-close to a rank-$3$ matrix. In this case we consider the product $M_{i_1i_2i_3i_4}x_{i_1}y_{i_2}$ for $x$ and $y$ random Gaussian vectors. This is
$$
(x^T \Sigma y) \Sigma/3 + 1/3(\Sigma x)\otimes(\Sigma y) + 1/3(\Sigma y)\otimes(\Sigma x) + (M-\Sigma)(x\otimes y).
$$
We note that the last term has mean square $\|M-\Sigma\|_F^2=O(\eta^2)$. The first term is $\Sigma$ times something that is $\Theta(1)$ with $90\%$ probability, and the middle two terms yield a rank at most $2$ matrix. The algorithm picks four random Gaussian vectors, $x,y,z,w$. We note that with constant probability the following hold:
\begin{itemize}
\item $|x^T \Sigma y|,|z^T \Sigma w| = \Theta(1)$.
\item $\|(M-\Sigma)(x\otimes y)\|_F, \|(M-\Sigma)(z\otimes w)\|_F = O(\eta)$.
\item Each of $\Sigma x,\Sigma y, \Sigma z, \Sigma w$ is at least $\Omega(\eta^{1/2})$-far from the span of the other three.
\item Each of $\Sigma x,\Sigma y, \Sigma z, \Sigma w$ has norm $O(1)$.
\end{itemize}
That the first two conditions hold with high constant probability is clear. That the third one does depends on the assumption that $\Sigma$ is not $\eta^{1/2}$-close to a rank $3$ matrix. This means that for any three dimensional subspace $V$ we have that $\|\mathrm{Pr}_{V^\perp}\circ \Sigma\|_F > \eta^{1/2}$. This means that for random Gaussian $x$, with high constant probability $\Sigma x$ is $\Omega(\eta^{1/2})$-far from $V$. The last condition holds with high constant probability since the expected squared norm of $\Sigma x$ is $\|\Sigma\|_F^2 = 1$.

In the following, we assume that the algorithm has picked $w,x,y,z$ so that the above hold. The algorithm then guesses $\eta$-approximations $A$ and $B$ to $x^T \Sigma y$ and $z^T \Sigma w$, respectively (we note that guessing uniform random numbers in $[-1,1]$ is correct with $\eta^2$ probability). Assuming that this guess is correct we compute the matrix $D=A M_{i_1i_2i_3i_4}z_{i_1}w_{i_2} - B M_{i_1i_2i_3i_4}x_{i_1}y_{i_2}$. In this case we have that $D$ is $O(\eta)$-close in Frobenius norm to
$$
L:=A/3((\Sigma w)\otimes(\Sigma z) + (\Sigma z)\otimes(\Sigma w))-B/3((\Sigma x)\otimes(\Sigma y) + (\Sigma y)\otimes(\Sigma x)).
$$
We note that $L$ is a rank $4$ matrix. Letting $D'$ be the closest rank $4$ approximation to $D$ (found via singular value decomposition), we note that $\|D'-L\|=O(\eta)$. Let $V$ be the span of $D'$. We claim that each of $\Sigma w, \Sigma x, \Sigma y, \Sigma z$ are $O(\eta^{1/2})$-close to $V$. In particular, if $u$ is a unit vector orthogonal to $V$ then $D'u=0$ by definition. However,
$$
Lu = A/3((u\cdot \Sigma w)(\Sigma z) + (u\cdot \Sigma z)(\Sigma w))-B/3((u\cdot\Sigma x)(\Sigma y) + (u\cdot\Sigma y)(\Sigma x)).
$$
Since each of $\Sigma w, \Sigma x, \Sigma y, \Sigma z$ are at least $\Omega(\eta^{1/2})$-far from the span of the others, this has size at least
$$
\Omega(\eta^{1/2}(|u\cdot \Sigma w|+|u\cdot \Sigma x|+|u\cdot \Sigma y|+|u\cdot \Sigma z|)).
$$
On the other hand, $|(L-D')u|\leq \|L-D'\|_F|u| = O(\eta)$. Therefore, $\Sigma x, \Sigma y, \Sigma z$, and $\Sigma w$ must be $O(\eta^{1/2})$-close to $v$.

The algorithm now guesses random vectors $s,t$ in $V$ with norm at most $O(1)$. We note that with $\poly(\eta)$ probability that these are within $\eta$ of the projections of $\Sigma x$ and $\Sigma y$ onto $V$, and thus are $O(\eta^{1/2})$-approximations of $\Sigma x$ and $\Sigma y$. If this holds, then we note that $\Sigma$ is within $O(\eta^{1/2})$ of
$$
(3/A)(M_{i_1i_2i_3i_4}x_{i_1}y_{i_2}-s_{i_3} t_{i_4}/3 - t_{i_3} s_{i_4}/3).
$$
The algorithm returns this guess, which is sufficiently close with $\poly(\eta)$ probability. This completes the proof.

\end{proof}

\section{Putting Everything Together}\label{combineSec}

We now have all the necessary tools and are prepared to prove Theorem \ref{mainTheorem}.
\begin{proof}
We first design an algorithm that with probability $\poly(\eps)$ returns a hypothesis (that is a mixture of two Gaussians) that is $\poly(\eps)$ close to $X$. Running this $\poly(1/\eps)$ times, it is likely that at least one trial is actually close, and running the algorithm from Lemma \ref{tournamentLemma} over these hypotheses produced will give an appropriate answer.

We let $\delta$ be some very small polynomial in $\eps$. We begin by guessing whether $\dtv(G_1,G_2)$ is larger than $1-\delta$. If so, we apply Theorem \ref{separatedCaseTheorem} to get a hypothesis.

Otherwise, we apply Proposition \ref{standardFormReductionProp}. We then note that applying $L$ to our samples constitutes getting $\eps'$-noisy samples from $X'$. Writing $X'=(G_1'+G_2')/2$ with $G_1'=N(\mu,I-\mu\mu^T+\Sigma)$ and $G_2'=N(-\mu,I-\mu\mu^T-\Sigma)$, we apply Corollary \ref{momentEstimationCorollary} to these samples to learn $\eta=\tilde O(\sqrt{\eps'})$-approximations to $\mu^{\otimes 6}$ and $\sym(6\Sigma^{\otimes 2}-4\mu^{\otimes 4})$. We then apply Proposition \ref{muEstimationProp} to obtain a $\poly(\eps)$-approximation $\hat\mu$ to $\mu$ (assuming our guesses work). Letting $M$ be one sixth of the difference of our approximation to $\sym(6\Sigma^{\otimes 2}-4\mu^{\otimes 4})$ plus $4\hat\mu^{\otimes 4}$, we get that $M$ is a $\poly(\eps)$-approximation to $\sym(\Sigma^{\otimes 2})$. Letting $C$ be a large constant multiple of $\log(1/\delta)$, then $M/C^2$ is a $\poly(\eps)$-approximation to $\sym((\Sigma/C)^{\otimes 2})$, but since $\|\Sigma/C\|_F\leq 1$ we can apply Proposition \ref{SigmaEstimationProp} and multiply the answer by $C$ to get (with $\poly(\eps)$ probability) a $\poly(\eps)$ approximation $\hat\Sigma$ to $\Sigma$. We then return an answer of $G_1'\approx N(\hat \mu,I-\hat\mu \hat\mu^T+\hat\Sigma)$ and $G_2' \approx N(-\hat \mu,I-\hat\mu \hat\mu^T-\hat\Sigma)$. We note by Fact \ref{GaussianTVFact} that $(G_1'+G_2')/2$ is correct to error $\poly(\eps)/\delta^2$ (which is still $\poly(\eps)$ if $\delta$ was originally taken to be a sufficiently small polynomial in $\eps$). Taking $X=L^{-1}(X')$ gives our result.

The above gives an algorithm that gives a $\poly(\eps)$-approximation with probability $\poly(\eps)$. Running this algorithm $\poly(1/\eps)$ times, we get a list of hypotheses $H_1,\ldots,H_m$ so that with high probability there exists an $i$ so that $\dtv(H_i,X)<\poly(\eps)$. Applying Lemma \ref{tournamentLemma} yields our final result.
\end{proof}

\section{Conclusions and Further Work}\label{conclusionsSec}

This resolves one of the major outstanding problems in computational robust statistics. There are three natural ways to try to extend this result, which we will briefly discuss.

Firstly, our result only covers equally weighted mixtures, while one might actually want to deal with arbitrary mixtures. This is a slight problem for us due to our use of the clustering result from \cite{clustering2}. In particular, their algorithm will only be polynomial time if the weights of the individual components are bounded away from 0. It should be possible to deal with this issue as we only really needed their algorithm to work if the components are separated in terms of their covariance (i.e. if $\|\Sigma_X^{-1/2} (\Sigma_1-\Sigma_2)\Sigma_X^{-1/2}\|_F$ is large), which may be possible with their techniques even for very unbalanced mixtures. Another slight technical issue is that after normalizing $X$, the formulas in Section \ref{momentsSec} would need to be modified for the unequal weights case. This problem should still be solvable using techniques along the lines of the ones we use, but doing so is not entirely trivial.

Secondly, our dependence on $\eps$ is rather poor. In particular, given $\eps$-noisy samples, we only guarantee and error of $\poly(\eps)$ in our final approximation. On the other hand, it might be reasonable to aim for an error of $\tilde O(\eps)$ (even though error $O(\eps)$ is information-theoretically possible, it seems unlikely as we do not know how to achieve this error efficiently for even a single Gaussian). Improving things in this way would likely require substantial new ideas. In particular, after reducing to the non-separated, normalized case, our algorithm proceeds by learning $X$ to small error in parameter distance. Such techniques are going to inherently have polynomial gaps due to an integrality gap. In particular, based on results of \cite{LearningMixturesWithoutNoise2}, we know that there are pairs of mixtures with parameter distance $\eps$ but with total variational distance approximately $\eps^6$, however for other pairs of mixtures the error in parameter distance is comparable to the error in total variational distance. Together this means that sampling a Gaussian up to an $\eps$ error rate is only sufficient to learn the parameters to error $O(\eps^{1/6})$. However, any generic algorithm that learns a mixture of Gaussians from an $\eps^{1/6}$-approximation to its parameters can only learn the Gaussian to error $\eps^{1/6}$. Thus, any algorithm attempting to obtain substantially better final error than this will need to find a way of dealing with this inconsistent relationship between distance and parameter distance.

Perhaps the most substantial generalization would be to cover the case of mixtures of $k$ Gaussians for any constant $k$. Here, due to lower bounds of \cite{SQ} in the statistical query model, it is likely that the running time would need to be at least polynomial in $d^k$, but such an algorithm might be plausible. However, doing this would require a somewhat substantial generalization of our techniques. Firstly, it is no longer sufficient to consider a binary nearby components versus far components. The algorithm will first want to split the components of $X$ into clusters where the Gaussians of each cluster are close to each other but far from the Gaussians in other clusters. Hopefully, techniques along the lines of those in \cite{clustering1} and \cite{clustering2} could be used to divide the samples into these clusters. But even having done this, one would now need to compute many moments of the (normalized) clusters and would need new algorithms for efficiently converting approximations to these moments into estimates of the individual components.

\end{document}